\documentclass[11pt]{article}

\usepackage[margin=1.05in]{geometry}
\usepackage{amsmath,amssymb,amsthm,mathtools}
\usepackage{booktabs,array}
\usepackage{enumitem}
\usepackage{microtype}
\usepackage[numbers,sort&compress]{natbib}
\usepackage[hidelinks]{hyperref}
\hypersetup{
  pdftitle={Unrestricted Boolean Multiplicative Complexity of Four-Term Binary Polynomial Multiplication},
  pdfauthor={Gregory Morse},
  pdfsubject={Unrestricted XOR--AND lower bounds for binary polynomial multiplication},
  pdfkeywords={multiplicative complexity, XOR--AND circuits, polynomial multiplication, Hasse jets}
}

\newtheorem{theorem}{Theorem}[section]
\newtheorem{lemma}[theorem]{Lemma}
\newtheorem{proposition}[theorem]{Proposition}

\newtheorem{definition}[theorem]{Definition}
\newtheorem{remark}[theorem]{Remark}
\newtheorem{conjecture}[theorem]{Conjecture}

\newcommand{\F}{\mathbf F_2}
\newcommand{\Aff}{\mathrm{Aff}}
\newcommand{\Span}{\operatorname{span}}
\newcommand{\Ann}{\operatorname{Ann}}
\newcommand{\Supp}{\operatorname{Supp}}
\newcommand{\rk}{\operatorname{rank}}
\newcommand{\MC}{\operatorname{MC}_{\wedge}}
\newcommand{\Mul}{\operatorname{Mul}}

\newcommand{\pihi}{\pi_{\ge 3}}
\newcommand{\Bool}{\mathcal B}

\title{\bf Unrestricted Boolean Multiplicative Complexity of\\
Four-Term Binary Polynomial Multiplication\\
\large Rational Places, Hasse Jets, and the Failure of Nonlinear Feedback}
\author{Gregory Morse\\
E{\"o}tv{\"o}s Lor{\'a}nd University}
\date{}

\begin{document}
\maketitle

\begin{abstract}
Classical lower bounds show that multiplying two degree-three polynomials
over $\mathbb F_2$ requires nine scalar products in bilinear or quadratic
models. They do not settle unrestricted Boolean multiplicative complexity:
an XOR--AND circuit may reuse nonlinear intermediate wires, and Boolean
equality is taken modulo $x_i^2=x_i$, so a multiplication can lower algebraic
degree. Let $\operatorname{Mul}_4:\mathbb F_2^8\to\mathbb F_2^7$ output the
seven coefficients of the product of two four-term binary polynomials. We
prove that its unrestricted XOR--AND multiplicative complexity is exactly
nine. This resolves, for a natural vector-valued quadratic function, the
Boyar--Find question of whether a quadratic-circuit lower bound can persist
against unrestricted nonlinear reuse. The proof is structural rather than
exhaustive. A useful purely quadratic prefix is forced onto the three rational
places of $\mathbb P^1(\mathbb F_2)$. In a hypothetical eight-AND circuit, the
unique non-useful gate must carry a cubic high part. Any useful continuation
then forces a rational tangent and exposes a first Hasse jet, while exterior
jet separation together with Boolean idempotence prevents the same defect
from exposing the second Hasse jet. The required useful suffix therefore
cannot exist. A complete Lean 4 formalization verifies the Boolean-ANF
semantics, the unrestricted circuit model, and the exact theorem; it uses no
project-specific axiom or native decision procedure. The same zero-defect flag
argument gives multiplicative complexity six for three-term multiplication,
and the method isolates the multi-defect obstruction for five terms.
\end{abstract}

\section{Introduction}

Multiplication of short polynomials over finite fields is a classical
test problem for multiplicative complexity.  Bilinear and quadratic
algorithms are tightly connected with Hankel matrices, linear-recurring
sequences, interpolation, finite-field places, and coding theory; see,
among others, Lempel--Seroussi--Winograd~\cite{LSW83},
Kaminski~\cite{Kaminski85}, Bshouty--Kaminski~\cite{BK90}, and
Averbuch--Bshouty--Kaminski~\cite{ABK92}.  For two degree-three
polynomials over $\F$, the tight value nine is classical
\cite{Kaminski85}.  Find--Peralta recover and exploit the corresponding
coding-theoretic lower bounds in the language of \emph{quadratic
Boolean circuits}, while their circuit search is restricted further to
symmetric bilinear circuits~\cite{FindPeralta19}.

There is, however, a different Boolean question.  Let the coefficients
$a_i,b_j$ be Boolean inputs and identify functions by their values on
$\F^8$, equivalently in the ANF quotient
\[
 \F[a_0,\ldots,a_3,b_0,\ldots,b_3]/(x_i^2-x_i).
\]
An unrestricted XOR--AND circuit may multiply arbitrary XOR
combinations of \emph{previous nonlinear wires}.  Such a multiplication
can create cubic or quartic terms and, because $x_i^2=x_i$, can also
contract higher-degree expressions back into degree two.  None of this
is present in a bilinear or quadratic circuit.

This distinction is not terminological.  Boyar--Find define
multiplicative complexity as the minimum number of AND gates in an
arbitrary XOR--AND circuit and ask whether, over finite fields and in
particular over $\F$, quadratic vector-valued Boolean functions can
always be computed optimally by quadratic circuits
\cite{BoyarFind18}.  They note that Mirwald--Schnorr proved quadratic
circuits optimal for one- and two-output quadratic Boolean functions,
but that the general $(n,m)$ case is open~\cite{MirwaldSchnorr92}.
The algebraic straight-line lower bounds of Bshouty--Kaminski
\cite{BshoutyKaminski06} concern computation of polynomial
coefficients as formal field expressions; they do not identify Boolean
functions modulo $x_i^2-x_i$ and therefore do not settle this question.

\subsection*{Four models that must not be conflated}

For the present function the relevant hierarchy is
\[
 \text{bilinear}\subseteq\text{quadratic Boolean}
 \subseteq\text{unrestricted Boolean},
\]
while algebraic straight-line complexity is a related but different
formal-polynomial model.  Concretely:

\begin{center}
\begin{tabular}{p{0.20\linewidth}p{0.34\linewidth}p{0.34\linewidth}}
\toprule
model & allowed multiplication & equality notion\\
\midrule
bilinear
 & $\ell(a)\,\ell'(b)$
 & polynomial/Boolean equality; no nonlinear reuse\\
quadratic Boolean
 & $\Sigma\Pi\Sigma$ XOR--AND; equivalently every gate computes a Boolean function of algebraic degree at most two
 & equality of Boolean functions\\
algebraic straight-line
 & arbitrary previous field expressions, possibly with division
 & formal polynomial/rational identity over the field\\
unrestricted Boolean
 & arbitrary previous XOR wires, including nonlinear intermediates
 & Boolean equality modulo $x_i^2=x_i$\\
\bottomrule
\end{tabular}
\end{center}

The classical nine-multiplication lower bound belongs to the first
polynomial/quadratic traditions in the appropriate formulations.  The
theorem of this paper is the last row:
\[
 \boxed{\text{nonlinear Boolean feedback still cannot beat nine ANDs}.}
\]

Our proof explains why.  The target contains exactly three rank-one
Hankel directions over $\F$, the rational places of
$\mathbb P^1(\F)$.  A hypothetical eight-AND computation has exactly
one non-useful quotient direction.  After normalization, that defect
can interact profitably with the target only as a rational tangent.
It can expose the first Hasse jet at that place, but not the second.
Since four suffix gates would all have to be useful, the circuit dies
after the first feedback step.

The proof is algebraic.  It has also been formalized end to end in Lean~4,
from Boolean functions and unrestricted circuit semantics through the
seven- and eight-gate exclusions and the explicit nine-gate upper bound.
The resulting kernel-checked theorem is the same unrestricted equality
stated here.  Separate small programs used during development remain as
regression tests, but neither their code nor their recorded outputs are a
logical premise of the theorem.  The pinned public formalization and
regression release is identified after the proof.

\paragraph{Contribution and scope.}
The new result is the unrestricted Boolean equality
\[
 \MC(\Mul_4)=9.
\]
An accompanying contribution is a complete formal proof of the exact
frontier
\[
 \bigl(\MC(\Mul_0),\MC(\Mul_1),\MC(\Mul_2),
        \MC(\Mul_3),\MC(\Mul_4)\bigr)=(0,1,3,6,9)
\]
in Lean~4~\cite{deMouraUllrich21}, using mathlib~\cite{Mathlib20}.
In particular, the formal development does not merely check the upper-bound
circuit or isolated finite cases: it exports the unrestricted lower bound
and the exact equality for $n=4$ in the same circuit model as the paper.
The value nine was previously known for restricted polynomial-
multiplication/quadratic models; what is proved here is that the
additional power of arbitrary Boolean nonlinear reuse does not lower
it.  Proposition~\ref{prop:smallbase} settles the elementary cases
of zero-, one-, and two-term multiplication directly, while
Proposition~\ref{prop:n3} supplies the three-term lifting with zero
defect budget.  To our knowledge,
$\Mul_4$ is the first explicit polynomial-multiplication instance in
this family where a \emph{nonzero} defect budget must be controlled
to lift a nontrivial quadratic-circuit lower bound to unrestricted
Boolean circuits.  The proof mechanism---rational-place normalization,
feedback annihilators, and Hasse-jet saturation---is intended to
transfer to multi-defect instances rather than merely reprove a small
bilinear rank.

\section{The Boolean circuit flag}\label{sec:model}

Let
\[
 \Bool=\F[a_0,\ldots,a_3,b_0,\ldots,b_3]/
          (x^2-x:x\in\{a_i,b_j\})
\]
be the Boolean ANF algebra.  XOR and constants are free.  An AND gate
takes two functions already available in the current $\F$-linear
wire space and adjoins their Boolean product.

\begin{definition}[unrestricted Boolean multiplicative complexity]
An \emph{unrestricted XOR--AND circuit} starts from the input variables
and constants, allows arbitrary XORs at zero cost, and at each AND gate
multiplies any two functions in the current wire span.  For a
vector-valued Boolean function $F$, define
\[
 \MC(F)=\min_C\#\{\text{AND gates of }C\},
\]
where the minimum ranges over unrestricted XOR--AND circuits computing
$F$.

Following Boyar--Find, a circuit is \emph{quadratic} when it is both
a $\Sigma\Pi\Sigma$ circuit and an XOR--AND circuit.  They state that
this is equivalent to requiring every gate of the circuit to compute
a Boolean function of algebraic degree at most two
\cite[Sec.~2]{BoyarFind18}.  A quadratic circuit is \emph{bilinear}
when, after partitioning the inputs into the $a$- and $b$-variables,
each input to an AND gate is a linear combination from one side of the
partition and the other input is a linear combination from the other
side~\cite[Sec.~2]{BoyarFind18}.
\end{definition}

The outputs of $\Mul_4$ are quadratic functions, but an unrestricted
circuit computing them is not assumed to remain quadratic internally.
This is the model distinction on which the theorem depends.

Write
\[
 E_s=\sum_{\substack{0\le i,j\le3\\i+j=s}}a_i b_j,
 \qquad 0\le s\le6,
\]
and
\[
 T=\Span\{E_0,\ldots,E_6\}.
\]
The seven $E_s$ are linearly independent.  Let $\Aff$ be the
nine-dimensional space spanned by $1$ and the eight input variables,
and put $W=\Aff+T$.

If $V$ is the current wire space, define its \emph{target dimension}
by
\[
 d_T(V)=\dim\bigl((V\cap W)/\Aff\bigr).
\]
An AND gate is \emph{useful} if $d_T$ rises by one, and
\emph{non-useful} otherwise.

\begin{lemma}[one gate buys at most one target dimension]\label{lem:one}
If $V'=V+\langle g\rangle$, then
\[
 d_T(V')-d_T(V)\in\{0,1\}.
\]
\end{lemma}

\begin{proof}
Since $\Aff\subseteq V\subseteq V'$, the increase equals
\[
 \dim(V'\cap W)-\dim(V\cap W)
 =\dim\bigl((V'\cap W)/(V\cap W)\bigr).
\]
The natural map
\[
 (V'\cap W)/(V\cap W)\longrightarrow V'/V
\]
is injective: its kernel is
\[
 \bigl((V'\cap W)\cap V\bigr)/(V\cap W)
 =(V\cap W)/(V\cap W)=0.
\]
Therefore the increase is at most $\dim(V'/V)\le1$, and it is
nonnegative.
\end{proof}

\subsection{The base frontier: zero through three terms}

For $n\geq 1$, let
\[
 \Mul_n:\F^{2n}\longrightarrow\F^{2n-1}
\]
send the coefficients of
\[
 A(z)=\sum_{i=0}^{n-1}a_i z^i,
 \qquad
 B(z)=\sum_{j=0}^{n-1}b_j z^j
\]
to the $2n-1$ coefficients
\[
 E_s=\sum_{i+j=s}a_i b_j,
 \qquad 0\leq s\leq 2n-2,
\]
of $A(z)B(z)$.  For completeness, let $\Mul_0$ be the unique map
between zero-dimensional vector spaces.

\begin{proposition}[elementary base cases]\label{prop:smallbase}
In the unrestricted Boolean model,
\[
 \boxed{\MC(\Mul_0)=0,\qquad
        \MC(\Mul_1)=1,\qquad
        \MC(\Mul_2)=3.}
\]
\end{proposition}

\begin{proof}
The value for $n=0$ is the stated empty-input convention.  For
$n\geq1$, let $\Aff_n$ be the span of the constant function and the
$2n$ input variables.  Starting from $\Aff_n$, each AND gate adjoins
at most one new coset modulo $\Aff_n$.  Hence a circuit with $r$ AND
gates has wire-space quotient dimension at most $r$ modulo $\Aff_n$.

The $2n-1$ target coordinates $E_s$ are linearly independent modulo
$\Aff_n$: the quadratic monomial $a_i b_j$ occurs in exactly the
coordinate $E_{i+j}$, so distinct coordinates have disjoint nonempty
quadratic monomial supports.  Every circuit for $\Mul_n$ therefore
uses at least $2n-1$ AND gates.  This gives the lower bounds one and
three for $n=1$ and $n=2$.

For $n=1$, the single product $a_0b_0$ attains the bound.  For $n=2$,
use
\[
 p_0=a_0b_0,\qquad p_1=a_1b_1,\qquad
 p_2=(a_0+a_1)(b_0+b_1).
\]
Then
\[
 E_0=p_0,\qquad E_1=p_0+p_1+p_2,\qquad E_2=p_1,
\]
so three AND gates suffice.
\end{proof}

The next case is the first for which the target-dimension lower bound
is not tight.  It is also the last case whose attempted improvement
has zero defect budget.

\subsubsection*{Zero-defect lifting for three terms}

Specialize the preceding definition to multiplication of
\[
 a_0+a_1z+a_2z^2
 \qquad\text{and}\qquad
 b_0+b_1z+b_2z^2.
\]

\begin{proposition}[the zero-defect case]\label{prop:n3}
The unrestricted Boolean multiplicative complexity of three-term
binary polynomial multiplication is
\[
 \boxed{\MC(\Mul_3)=6}.
\]
\end{proposition}

\begin{proof}
The five output coefficients are linearly independent modulo affine
functions: their degree-two monomial supports are disjoint, and no
nonzero quadratic target combination is affine.  Hence the target
quotient has dimension five.  If a five-AND circuit computed
$\Mul_3$, Lemma~\ref{lem:one} would force all five gates to be useful.  Let $W_3$ be affine functions plus the
five target coordinates.  Starting from $V_0\subset W_3$, usefulness
and $\dim(V_j/V_{j-1})=1$ imply inductively that $V_j\subset W_3$:
if the new gate lay outside $W_3$, the intersection with $W_3$ could
not increase.  Thus every gate computes a Boolean function of degree at most two.
By Boyar--Find's equivalence this is a quadratic circuit
\cite[Sec.~2]{BoyarFind18}.  Find--Peralta's lower-bound section states
that any quadratic Boolean circuit for $n$-term polynomial
multiplication induces the corresponding linear code; their table gives
the lower bound six for $n=3$
\cite[Sec.~5, Table~3]{FindPeralta19}.  Contradiction.

Six ANDs suffice.  With
\[
\begin{array}{lll}
p_0=a_0b_0,&p_1=a_1b_1,&
p_2=(a_0+a_1)(b_0+b_1),\\
p_3=a_2b_2,&
p_4=(a_0+a_2)(b_0+b_2),&
p_5=(a_1+a_2)(b_1+b_2),
\end{array}
\]
the five coefficients are
\[
\begin{aligned}
E_0&=p_0,\\
E_1&=p_0+p_1+p_2,\\
E_2&=p_0+p_1+p_3+p_4,\\
E_3&=p_1+p_3+p_5,\\
E_4&=p_3.
\end{aligned}
\]
\end{proof}

Combining Propositions~\ref{prop:smallbase} and~\ref{prop:n3}, the
complete unrestricted base frontier is
\[
 \bigl(\MC(\Mul_0),\MC(\Mul_1),\MC(\Mul_2),\MC(\Mul_3)\bigr)
   =(0,1,3,6).
\]
Proposition~\ref{prop:n3} is the $m>2$ Boolean lifting with \emph{zero} defect
budget: five gates for five target dimensions leave no room for a
non-useful direction.  The four-term case is the first instance in
this family where a circuit below the quadratic optimum has a
nonzero defect budget and nonlinear feedback must be ruled out
structurally.

Consequently, any eight-AND circuit computing $\Mul_4$ has exactly
one non-useful gate.

\subsection{Exterior top-degree bookkeeping}

Let
\[
 L=\Span\{a_0,\ldots,a_3,b_0,\ldots,b_3\}.
\]
For highest homogeneous components we identify a squarefree monomial
$x_{i_1}\cdots x_{i_d}$ with
$x_{i_1}\wedge\cdots\wedge x_{i_d}\in\Lambda^dL$.
If two Boolean ANFs are multiplied, their top component is the
ordinary exterior product; repeated variables produce lower-degree
``contraction'' terms in the Boolean algebra.

There is one characteristic-two point that is easy to misread.

\begin{remark}[wedge square versus Pfaffian in characteristic two]\label{rem:square}
For every $G\in\Lambda^2L$,
\[
 G\wedge G=0.
\]
Indeed, every monomial square is zero in the exterior algebra and each
cross term occurs twice.  Thus the ordinary wedge square carries
\emph{no} decomposability information in characteristic two.  The
Pfaffian, or equivalently the divided-square quadratic form, is a
different operation.  We never use $G\wedge G$ as a rank detector.
\end{remark}

For a two-form $\omega$, $\Supp(\omega)$ denotes the support of its
alternating matrix, equivalently the smallest subspace $U\le L$ with
$\omega\in\Lambda^2U$.

\begin{lemma}[rank-four support]\label{lem:support}
If $\omega$ has alternating rank four,
$U\le L$ has dimension two, and
$\omega\in U\wedge L$, then
\[
 U\subseteq\Supp(\omega).
\]
\end{lemma}

\begin{proof}
Choose a basis $u_1,u_2$ of $U$ and write
\[
 \omega=u_1\wedge x+u_2\wedge y+\lambda u_1\wedge u_2.
\]
Rank four implies that the images of $x,y$ modulo $U$ are independent.
Hence both $u_1$ and $u_2$ occur in the four-dimensional support.
\end{proof}

\section{Hankel target geometry}\label{sec:hankel}

For
\[
 t=\sum_{s=0}^6c_sE_s\in T
\]
define the $4\times4$ Hankel matrix
\[
 H(t)=
 \begin{pmatrix}
 c_0&c_1&c_2&c_3\\
 c_1&c_2&c_3&c_4\\
 c_2&c_3&c_4&c_5\\
 c_3&c_4&c_5&c_6
 \end{pmatrix}.
\]
In the decomposition
$L=A\oplus B$,
$A=\Span\{a_i\}$, $B=\Span\{b_j\}$, the alternating matrix of $t$ is
\[
 M(t)=
 \begin{pmatrix}
 0&H(t)\\
 H(t)^T&0
 \end{pmatrix}.
\]
Therefore
\begin{equation}\label{eq:rankdouble}
 \rk M(t)=2\,\rk H(t).
\end{equation}
This block identity, rather than $H+H^T$, is the source of the factor
two in characteristic two.

\subsection{The three rational places}

Set
\[
 r_0=E_0=a_0b_0,\qquad
 r_\infty=E_6=a_3b_3,\qquad
 r_1=\sum_{s=0}^6E_s=A(1)B(1),
\]
and
\[
 R=\Span\{r_0,r_1,r_\infty\}.
\]
These are the three $\F$-rational points of the rank-one Hankel curve.

\begin{lemma}[rank-one target forms]\label{lem:rankone}
The nonzero decomposable elements of $T$ are exactly
\[
 r_0,\quad r_1,\quad r_\infty.
\]
\end{lemma}

\begin{proof}
Let a decomposable target form be
$(\alpha+\beta)\wedge(\gamma+\delta)$ with
$\alpha,\gamma\in A$ and $\beta,\delta\in B$.
The absence of $A\wedge A$ and $B\wedge B$ components forces
$\alpha,\gamma$ to be dependent and $\beta,\delta$ to be dependent.
Thus the cross matrix has rank at most one.

Conversely, a nonzero rank-one Hankel matrix is symmetric and can be
written $vv^T$.  If $k$ is the first nonzero index of $v$, the Hankel
condition forces $k=0$ or $k=3$.  If $k=3$ one gets $r_\infty$.
If $k=0$, the recurrence forced by equality along anti-diagonals gives
either $v=(1,0,0,0)$, yielding $r_0$, or
$v=(1,1,1,1)$, yielding $r_1$.
\end{proof}

The rank-one classification is classical; its role below is to
control which target directions can first enter an XOR--AND flag.

\subsection{Useful purely quadratic prefixes}

\begin{proposition}[rational-place prefix theorem]\label{prop:prefix}
Let $S\le R$ and $V=\Aff+S$.  Every useful AND extension of $V$ adds
one of the missing rational-place directions
$r_0,r_1,r_\infty$.

Consequently the all-useful target-prefix states have cardinalities
\[
 3,\quad3,\quad1,\quad0
\]
after one, two, three, and four useful gates respectively.
\end{proposition}

\begin{proof}
Write two available wires as
\[
 u=\alpha+\ell+Q,\qquad
 v=\beta+m+C,
\]
with $\ell,m\in L$ and $Q,C\in S$.
If the product is useful without creating a new high-degree direction,
its quartic top satisfies
\[
 Q\wedge C=0.
\]
In an adapted basis for the three rational places,
the three four-forms
$r_0\wedge r_1$, $r_0\wedge r_\infty$,
$r_1\wedge r_\infty$ are independent.  Hence $Q,C$ are dependent:
\[
 Q=\varepsilon G,\qquad C=\delta G
\]
for some $G\in S$.

Put $N=\delta\ell+\varepsilon m$.  The cubic top is $N\wedge G$.
If $G$ contains at least two rational-place components, its alternating
rank is at least four, so
$N\mapsto N\wedge G$ is injective and vanishing of the cubic part
forces $N=0$.  If $G=r_\theta$ is a single rational place, vanishing
forces $N\in P_\theta=\Span\{A_\theta,B_\theta\}$; multiplying a vector
of $P_\theta$ by $A_\theta B_\theta$ produces only a multiple of
$r_\theta$ in degree two.  Finally $G^2=G$ in the Boolean algebra.
Thus all nonlinear contributions reduce into $S$.

Modulo $S$, any new target direction therefore comes from a product of
two affine forms.  By Lemma~\ref{lem:rankone}, it is one of
$r_0,r_1,r_\infty$.  The state counts are then the counts of
one-, two-, and three-element subsets of a three-element set.
\end{proof}

\section{Rank-two target forms and the degree-two place}\label{sec:ranktwo}

We need one further finite piece of Hankel geometry.  It is small
enough to display completely.

\begin{lemma}[binary rank-two Hankels]\label{lem:ranktwo}
The coefficient words $c_0c_1\cdots c_6$ for which
$\rk H(c)\le2$ are exactly
\[
\begin{array}{c|l}
c_0c_1 & c_0c_1\cdots c_6\\ \hline
00&0000000,\ 0000010,\ 0000001,\ 0000011\\
01&0100000,\ 0101010,\ 0110110,\ 0111111\\
10&1000000,\ 1000001,\ 1010101,\ 1011011\\
11&1100000,\ 1101101,\ 1111110,\ 1111111.
\end{array}
\]
\end{lemma}

\begin{proof}
The four cases are obtained by row-span elimination, not by a search.
For example, if $c_0=c_1=0$, successive pivots force
$c_2=c_3=c_4=0$, leaving $c_5,c_6$ free.
If $(c_0,c_1)=(0,1)$, the first two rows are independent and rows
three and four lying in their span force
\[
 c_4=c_2,\qquad
 c_5=c_2+c_3+c_2c_3,\qquad
 c_6=c_2c_3.
\]
The four choices of $(c_2,c_3)$ give the second row of the table.
For $(1,0)$, either the second row vanishes, yielding
$1000000,1000001$, or independence gives
\[
 c_2=1,\quad c_4=1+c_3,\quad c_5=c_3,\quad c_6=1.
\]
For $(1,1)$, either the first two rows coincide, yielding
$1111110,1111111$, or independence gives
\[
 c_2=0,\quad c_4=c_3,\quad c_5=0,\quad c_6=c_3.
\]
\end{proof}

Among the fifteen nonzero forms, six lie in $R$, six are the two
first-order tangent forms at each rational place,
\[
\begin{array}{c|c}
0&E_1,\ E_0+E_1\\
1&E_1+E_3+E_5,\ E_0+E_2+E_4+E_6\\
\infty&E_5,\ E_5+E_6,
\end{array}
\]
and the remaining three are
\[
 0110110,\qquad1101101,\qquad1011011.
\]
They span a two-dimensional plane $D_*$ satisfying
\[
 c_{s+2}=c_{s+1}+c_s,
\]
the recurrence with minimal polynomial $x^2+x+1$.
Thus $D_*$ is the degree-two closed place over $\F$.

Let $\alpha\in\mathbf F_4$ satisfy $\alpha^2+\alpha+1=0$.  Write
\[
 A(\alpha)=u_A+v_A\alpha,\qquad
 u_A=a_0+a_2+a_3,\quad v_A=a_1+a_2,
\]
and similarly $u_B,v_B$.  Put
\[
 K_{\mathrm{deg}\,2}
 =\Span\{u_A,v_A,u_B,v_B\}.
\]
Then
\[
 T\cap\Lambda^2K_{\mathrm{deg}\,2}=D_*.
\]
In the bases $(u_A,v_A)$ and $(u_B,v_B)$, a nonzero member of $D_*$
has cross matrix
\[
 \begin{pmatrix}x&y\\y&x+y\end{pmatrix},
\qquad (x,y)\ne(0,0),
\]
whose determinant is the $\mathbf F_4/\F$ norm
$x^2+xy+y^2=1$.  Hence every nonzero member of $D_*$ has alternating
rank four and $D_*$ misses the Grassmannian.

For later support arguments let
\[
 K_0=\Span\{a_0,a_1,b_0,b_1\},\qquad
 K_\infty=\Span\{a_3,a_2,b_3,b_2\},
\]
and
\[
 K_1=\Span\{A(1),A^{(1)}(1),B(1),B^{(1)}(1)\},
\]
where $A^{(1)}(1)=a_1+a_3$ and similarly for $B$.
The six tangent forms have support $K_0,K_1,K_\infty$ respectively,
while the degree-two forms have support $K_{\mathrm{deg}\,2}$.

\section{Flag normalization}\label{sec:normalization}

The next lemma is the combinatorial step on which the later exterior
argument rests.

\begin{lemma}[first-entry replacement]\label{lem:firstentry}
Suppose $t\in T$ first enters the flag at gate $j$:
$t\in V_j$ but $t\notin V_{j-1}$.  Then gate $j$ may be replaced by
any direct one-AND realization of $t$ without changing the subspace
$V_j$ or any later realizable subspace.

In particular, if $t=r_\theta$, gate $j$ may be replaced by
$A_\theta B_\theta$.
\end{lemma}

\begin{proof}
There is $v\in V_{j-1}$ with
\[
 g_j+v=t.
\]
Hence
\[
 \langle V_{j-1},g_j\rangle
 =\langle V_{j-1},t\rangle.
\]
All later gate inputs are free XOR combinations of the current basis
wires, so they may be rewritten in the new basis with no additional
AND gates.
\end{proof}

\begin{lemma}[normalized eight-gate form]\label{lem:normalized}
If an eight-AND circuit computes $\Mul_4$, there is an eight-AND
circuit with the same property whose first three gates are
$r_0,r_1,r_\infty$, whose fourth gate is the unique non-useful gate,
and whose gates five through eight are all useful.
\end{lemma}

\begin{proof}
All three rational-place directions belong to the final target space,
so each first enters the flag somewhere.  Apply
Lemma~\ref{lem:firstentry} to replace those first-entry gates by their
direct products.  Since each such product depends only on the original
inputs, commute the gate to the front.  Moving it earlier does not
invalidate any old gate: gates that formerly preceded it did not use
it, and later gates may be XOR-rewired as above.  Repeat for all three
places.

The first three gates are independent target directions and therefore
useful.  The next non-place gate is now a legal extension of
$\Aff+R$.  Proposition~\ref{prop:prefix} says it cannot be useful, so
it is the unique non-useful gate.  Lemma~\ref{lem:one} and
$\dim T=7$ then force every remaining gate to be useful.
\end{proof}

Call gate four the \emph{seed} and write its output as $g$.

\begin{lemma}[the seed has nonzero high part]\label{lem:highnonzero}
The seed of a normalized hypothetical eight-AND circuit has
$\pihi(g)\ne0$.
\end{lemma}

\begin{proof}
If the seed had degree at most two, suppose a later gate were the first
gate with nonzero degree-$\ge3$ part.  Before that gate the whole state
would have degree at most two.  Any combination involving the new
high component could not lie in $\Aff+T$, so that first high-degree
gate would be non-useful.  This would be a second non-useful gate,
impossible.

Hence every gate in the eight-gate circuit would compute a Boolean
function of degree at most two, so Boyar--Find's equivalence makes the
whole circuit quadratic~\cite[Sec.~2]{BoyarFind18}.  Find--Peralta
state in their lower-bound section that any quadratic Boolean circuit
for $n$-term polynomial multiplication induces the relevant linear
code, and their table gives the lower bound nine for $n=4$
\cite[Sec.~5, Table~3]{FindPeralta19}.  This is the Boolean-quadratic
formulation of the classical Kaminski lower bound~\cite{Kaminski85},
and rules out eight gates.
\end{proof}

\section{A completion-capable seed is cubic}\label{sec:quartic}

The seed is a product of two wires from $\Aff+R$.  Such a product can
have degree four.  We now rule out that possibility algebraically.

Write its two factors as
\[
 u=\alpha+\ell+Q,\qquad
 v=\beta+m+C,
\qquad Q,C\in R.
\]
The quartic part is $Q\wedge C$.

\begin{proposition}[quartic exclusion]\label{prop:quartic}
If the normalized seed admits even one useful suffix gate, then
\[
 \pi_4(g)=0.
\]
Consequently the nonzero high part of the seed is purely cubic.
\end{proposition}

\begin{proof}
Assume $Q\wedge C\ne0$.  Then $Q,C$ are independent and span a
two-plane $\mathcal P\le R$.

\smallskip
\emph{Low--low child.}
First suppose a useful gate five multiplies two wires in $\Aff+R$.
Its quartic part must equal $Q\wedge C$, otherwise the new high
direction cannot cancel against the one-dimensional seed state.
Changing the basis of the two factors by $\mathrm{GL}_2(\F)$ changes
their product only by available wires, because
$u(u+v)=uv+u$.  We may therefore align the two quadratic bases.
If the differences of the corresponding linear parts are $x,y$, the
cubic cancellation equation is
\begin{equation}\label{eq:quarticcubic}
 x\wedge C+y\wedge Q=0.
\end{equation}

Up to the $S_3$ action on the rational places, two-planes in $R$ have
three orbits:
\[
\begin{array}{c|c}
\mathcal P&\text{consequence of \eqref{eq:quarticcubic}}\\ \hline
\langle r_0,r_1\rangle
 &x\in P_1,\ y\in P_0\\
\langle r_0,r_1+r_\infty\rangle
 &x=0,\ y\in P_0\\
\langle r_0+r_1,r_0+r_\infty\rangle
 &x=y=0.
\end{array}
\tag{Q4}
\]
This follows by comparing the direct cubic components
$L\wedge r_0$, $L\wedge r_1$, $L\wedge r_\infty$.
In the first row, for example,
$x\wedge r_1+y\wedge r_0=0$ forces both summands to vanish separately.

In all three rows the degree-two contraction of
$xC+yQ$ lies in $R$: if $x\in P_\theta$, then
$x\,r_\theta$ has no cubic part and its quadratic part is a multiple
of $r_\theta$.  Thus, modulo $R$, the quadratic difference of the two
products belongs to $U\wedge L$ with $U=\langle x,y\rangle$.

If $\dim U\le1$, every nonzero form in $U\wedge L$ is decomposable;
a target form of this kind lies in $R$ by Lemma~\ref{lem:rankone}.
If $\dim U=2$, only the first row of (Q4) is possible, with nonzero
vectors from both $P_0$ and $P_1$.  Any target form
$\omega\in U\wedge L$ has alternating rank at most four and hence
Hankel rank at most two.  If it is outside $R$, the list in
Lemma~\ref{lem:ranktwo} says that its support is one of
$K_0,K_1,K_\infty,K_{\mathrm{deg}\,2}$.  By
Lemma~\ref{lem:support}, $U$ would be contained in that support.
But no one of these four support spaces contains simultaneously a
nonzero vector of $P_0$ and a nonzero vector of $P_1$.
Contradiction.

\smallskip
\emph{Child using the seed.}
Now suppose gate five uses $g$.  Modulo the current state, a product
with two $g$-containing factors can be replaced by one with exactly
one such factor, again using $U(U+V)=UV+U$.  Write
\[
 U=g+a,\qquad a,c\in\Aff+R.
\]
Toggling $c$ by $1$ changes the product by the existing wire $U$.
Thus a useful extension has a representative
\[
 F=Uc\in\Aff+T,\qquad F\notin\Aff+R.
\]
Boolean idempotence gives
\begin{equation}\label{eq:idempquartic}
 UF=F,\qquad Fc=F.
\end{equation}

Let $D=\pi_2(c)\in R$.  Degree six in $Uc$ gives
$(Q\wedge C)\wedge D=0$, hence $D\in\mathcal P$.
If $D=0$, degree three of $Fc=F$ makes the nonzero target quadratic
part $F_2$ annihilate the linear part of $c$.  Since
$F_2\notin R$, it has Hankel rank at least two, so wedge by $F_2$ is
injective on vectors.  Thus $c$ is constant, and $Fc=F$ forces
$c=1$, giving $F=U$, which still has the seed high part.  Hence
$D\ne0$.

Degree four in $Fc=F$ gives
\[
 F_2\wedge D=0.
\]
For nonzero $D\in R$, a target annihilator outside $R$ exists only
when $D$ is a singleton rational place.  Up to $S_3$, take $D=r_0$;
then
\[
 F_2=E_1+\varepsilon E_0.
\]
The plane $\mathcal P$ contains $r_0$, and under the stabilizer of
$r_0$ there are two possibilities:
\[
 \mathcal P_A=\langle r_0,r_1\rangle,\qquad
 \mathcal P_B=\langle r_0,r_1+r_\infty\rangle.
\]

Put
\[
 x=a_0,\quad y=b_0,\quad u=a_1,\quad v=b_1,
\]
and
\[
 \bar A=a_1+a_2+a_3,\qquad
 \bar B=b_1+b_2+b_3.
\]
From the degree-three part of $Fc=F$ one obtains
\[
 c=\delta+\rho x+\sigma y+xy.
\]
Indeed any other linear variable would create an uncancelled cubic
monomial from $(uy+xv+\varepsilon xy)\,c$.  The coefficient of $xy$
is one, so $c$ has odd truth-table parity on the four $(x,y)$ slices.
Moreover $F=0$ wherever $c=0$, while the complementary linear parts
of $F$ on distinct slices differ by
$\Delta y\,u+\Delta x\,v$ and therefore cannot vanish on two
different slices.  Hence $c=1$ on exactly three slices.

Choose a factor basis for the seed:
\[
 g=(\ell+r_0)(m+Z),\qquad
 Z=
 \begin{cases}
 r_1,&\mathcal P_A,\\
 r_1+r_\infty,&\mathcal P_B.
 \end{cases}
\]
Modulo affine terms, on a fixed $(x,y)$ slice the second factor has
quadratic complementary part
\[
 Q_A=\bar A\bar B
\quad\text{or}\quad
 Q_B=\bar A\bar B+a_3b_3,
\]
and slice-varying linear part
$x\bar B+y\bar A$.

For $\mathcal P_B$, $Q_B$ has alternating rank four.  The cubic
part of $U_s=F_s$ on any of the three $c=1$ slices therefore forces
the complementary linear part of the first factor to vanish.
Quadratic matching then forces the scalar multiplying $Q_B$ to equal
the $r_1$ coefficient of $a$ and forces the $r_\infty$ coefficient
to be the same.  The slice-varying
$x\bar B+y\bar A$ terms cancel, so the complementary linear part of
$U_s$ is constant on the three slices.  This contradicts the two
independent differences $u,v$ of $F_s$.

For $\mathcal P_A$, write the complementary first-factor part as
$p\bar A+q\bar B$.  If it is zero, the same constant-slice
contradiction follows.  If it is nonzero, the complementary quadratic
equation first forces the $r_\infty$ coefficient of $a$ to vanish:
otherwise one obtains either equality of two decomposable forms with
disjoint support, or the nondecomposable rank-four form
$\bar A\bar B+a_3b_3$.  The same equation then forces the other
complementary linear factor into
$\Span\{\bar A,\bar B\}$.  Hence every slice-varying complementary
linear term of $U_s$ lies in
$\Span\{\bar A,\bar B\}$, which contains neither $u$ nor $v$.
Again the three-slice differences of $F_s$ give a contradiction.

Thus neither a low--low nor a $g$-using useful child exists when
$\pi_4(g)\ne0$.
\end{proof}

\section{Cubic seeds and feedback annihilators}\label{sec:feedback}

By Proposition~\ref{prop:quartic} and
Lemma~\ref{lem:highnonzero}, a completion-capable seed has a nonzero
cubic high part.

\begin{lemma}[low-product normal form]\label{lem:lownf}
Let $p$ be a product of two wires from $\Aff+R$, with nonzero cubic
part and zero quartic part.  Modulo $\Aff+R$ there exist
$G\in R\setminus\{0\}$ and $N,z\in L$ such that
\[
 \pihi(p)=N\wedge G
\]
and
\begin{equation}\label{eq:nfquad}
 \pi_2(p)\equiv z\wedge N+\kappa_G(N)\pmod R,
\end{equation}
where $\kappa_G(N)=\pi_2(NG)$ is the Boolean degree-lowering
contraction.
\end{lemma}

\begin{proof}
With notation $Q,C\in R$ as before, zero quartic part gives
$Q=\varepsilon G$, $C=\delta G$ for a common $G$.
Put $N=\delta\ell+\varepsilon m$.  Choose $z=\ell$ in the cases
$(\varepsilon,\delta)=(1,0),(1,1)$ and $z=m$ in the case $(0,1)$.
Then $\ell\wedge m=z\wedge N$.  Constants, $G^2=G$, and constant
multiples of $G$ lie in $\Aff+R$, leaving \eqref{eq:nfquad}.
\end{proof}

For a cubic $h$ define its target annihilator
\[
 \Ann_T(h)=\{t\in T:h\wedge t=0\}.
\]

\subsection{A direct annihilator classification}

The rational-place cubic spaces
\[
 I_\theta=L\wedge r_\theta
\]
form a direct sum:
\[
 I_0\oplus I_1\oplus I_\infty.
\]
Indeed, in a basis containing
$A_\theta,B_\theta$ for the three places, a cubic basis monomial
cannot contain two distinct rational-place pairs.

The following lemma is the exterior part of the feedback argument.

\begin{lemma}[annihilator dichotomy]\label{lem:anndichotomy}
Let
\[
 0\ne h=M\wedge G,\qquad 0\ne G\in R.
\]
If $\Ann_T(h)$ contains a direction outside $R$, then either
\begin{enumerate}[label=(\roman*)]
\item $h$ is anchored at a single rational place and
      $\dim\Ann_T(h)=3$, or
\item $h$ is anchored at a single rational place and
      $\dim\Ann_T(h)=2$.
\end{enumerate}
The second case is the only way an annihilator of dimension at most
two can meet $T\setminus R$.
\end{lemma}

\begin{proof}
The calculation is most transparent by the weight of $G$ in the
basis $r_0,r_1,r_\infty$.

If $G=r_0+r_1+r_\infty$, direct comparison of the three components
$I_\theta$ gives an annihilator of dimension at most two, entirely
contained in $R$.

If $G=r_\alpha+r_\beta$, an annihilator direction outside $R$ exists
only when $M\in P_\alpha\cup P_\beta$.  In that event one of the two
cubic components vanishes and
$h$ collapses to the other single anchor; the target annihilator has
dimension three.

The remaining case is $G=r_\theta$, already anchored.  For
$\theta=0$, $E_0$ and $E_1$ always annihilate
$M\wedge E_0$, so the annihilator has dimension at least two and at
most three.  The other places are equivalent under
$PGL_2(\F)\cong S_3$.

For completeness, the seven-column map being row-reduced is
\[
 T\longrightarrow\Lambda^5L,\qquad
 t\longmapsto M\wedge G\wedge t.
\]
In the adapted basis
\[
 e_0=A(0),\ e_1=A(1),\ e_\infty=A(\infty),\ e_*=a_1
\]
and analogously on the $B$ side, the cross matrix of
$t=\sum c_sE_s$ is
\[
\begin{pmatrix}
c_0+c_4&c_2+c_4&c_2+c_3+c_4+c_5&c_1+c_2+c_3+c_4\\
c_2+c_4&c_4&c_4+c_5&c_3+c_4\\
c_2+c_3+c_4+c_5&c_4+c_5&c_4+c_6&c_3+c_5\\
c_1+c_2+c_3+c_4&c_3+c_4&c_3+c_5&c_2+c_4
\end{pmatrix}.
\]
Substitution gives exactly the three cases above.
\end{proof}

Once a single anchor is known, the dimension-three case has a more
conceptual description.

\begin{lemma}[tail-place classification]\label{lem:tail}
Let $h=M\wedge E_0\ne0$.  Modulo
$P_0=\Span\{a_0,b_0\}$,
\[
 \dim\Ann_T(h)=3
\]
if and only if $[M]$ is a nonzero vector in one of
\[
\begin{aligned}
 P'_0&=\Span\{a_1,b_1\},\\
 P'_\infty&=\Span\{a_3,b_3\},\\
 P'_1&=\Span\{\bar A,\bar B\},
 \qquad
 \bar A=a_1+a_2+a_3,\quad
 \bar B=b_1+b_2+b_3.
\end{aligned}
\]
The corresponding annihilators are
\[
\begin{array}{c|c}
[M]\text{ plane}&\Ann_T(h)\\ \hline
P'_0&\langle E_0,E_1,E_2\rangle\\
P'_\infty&\langle E_0,E_1,E_6\rangle\\
P'_1&\langle E_0,E_1,E_2+E_3+E_4+E_5+E_6\rangle.
\end{array}
\]
All other nonzero classes have annihilator
$\langle E_0,E_1\rangle$.
\end{lemma}

\begin{proof}
Wedging first by $E_0=a_0b_0$ kills every term containing $a_0$ or
$b_0$.  The remaining target is exactly the five-dimensional
coefficient space for the product of the three-term tails
\[
 a_1+a_2z+a_3z^2,\qquad b_1+b_2z+b_3z^2.
\]
If a nonzero tail two-form $t$ satisfies $M\wedge t=0$, then either
$t$ has alternating rank two and $M$ lies in its support, or wedge by
$t$ is injective on vectors.  By Lemma~\ref{lem:rankone}, applied to
the $3\times3$ tail Hankel matrix, the only nonzero decomposable tail
targets are the three rational places
\[
 a_1b_1,\qquad a_3b_3,\qquad \bar A\bar B.
\]
Their support planes are exactly $P'_0,P'_\infty,P'_1$, which are
pairwise disjoint.  This gives the table.
\end{proof}

\subsection{A useful gate cannot use the seed}

We now analyze a hypothetical useful gate five that uses $g$.

First note that after anchoring at $0$ we may choose the generator of
the seed coset in the form
\begin{equation}\label{eq:seedrep}
 g=M(z+E_0)+b,\qquad b\in\Aff+R.
\end{equation}
Indeed, in Lemma~\ref{lem:lownf}, if the common $G$ contains a second
place $r_\phi$, the direct sum of the $I_\theta$ forces
$M\in P_\phi$, so $Mr_\phi$ is purely quadratic and lies in $R$.
Thus the extra place can be removed modulo $\Aff+R$.

\begin{proposition}[no $g$-using first feedback]\label{prop:nogfirst}
No useful gate five can use the seed $g$.
\end{proposition}

\begin{proof}
As in the quartic proof, a product with two $g$-containing factors is
equivalent modulo the state to one with exactly one.  Write
\[
 U=g+a,\qquad a,c\in\Aff+R.
\]
After toggling $c$ by one if necessary, a useful extension has a
low representative
\[
 F=Uc\in\Aff+T,\qquad F\notin\Aff+R.
\]
Again
\[
 UF=F,\qquad Fc=F.
\]
Degree five in the first identity gives
\[
 F_2\in\Ann_T(h)\setminus R.
\]
Lemma~\ref{lem:anndichotomy} reduces us to a single rational anchor.

Take that anchor to be $0$.  If $\dim\Ann_T(h)=2$, then
\[
 \Ann_T(h)=\langle E_0,E_1\rangle,
 \qquad
 F_2=E_1+\varepsilon E_0.
\]
If the annihilator has dimension three, Lemma~\ref{lem:tail} shows
that every outside-$R$ coset either has the same first-jet form
$E_1+\varepsilon E_0$, or, only in the $P'_0$ row, has a nonzero
$E_2$ coefficient.

We dispose first of an $E_2$ coefficient.  Write
\[
 F_2=E_2+\alpha E_1+\beta E_0,\qquad
 c=\delta+L+C,\quad C\in R.
\]
Degree four of $Fc=F$ gives $F_2\wedge C=0$.  On the columns
$E_0,E_6,r_1$, the three rows
\[
 a_0a_1b_0b_1,\quad
 a_2a_3b_0b_1,\quad
 a_2a_3b_0b_3
\]
give the matrix
\[
\begin{pmatrix}
1&0&1+\beta\\
0&0&1\\
0&1&1
\end{pmatrix},
\]
whose determinant is one.  Hence $C=0$.
The leading $3\times3$ minor of the Hankel matrix of $F_2$ is also
nonsingular, so $\rk H(F_2)=3$ and
$L\mapsto F_2\wedge L$ is injective.  Degree three gives $L=0$.
Thus $c$ is constant, $Fc=F$ forces $c=1$, and $F=U$ still has the
nonzero seed high part.  Contradiction.

It remains to rule out
\[
 F_2=E_1+\varepsilon E_0.
\]
Write
\[
 c=\delta+L+C,\qquad C\in R.
\]
The degree-four part of $Fc=F$ is $F_2\wedge C=0$.  Directly on
$R=\langle E_0,E_6,r_1\rangle$ its kernel is
$\langle E_0\rangle$.  If $C=0$, then the degree-three part gives
$F_2\wedge L=0$; since $F_2$ has alternating rank four, wedge by
$F_2$ is injective on vectors, so $L=0$ and $c$ is constant.  The
identity $Fc=F\ne0$ then forces $c=1$, giving $F=U$ and retaining the
nonzero seed high part.  Therefore
\[
 C=E_0.
\]

Put
\[
 x=a_0,\quad y=b_0,\quad u=a_1,\quad v=b_1.
\]
Expanding the degree-three part of $Fc=F$ now forces the remaining
linear part of $c$ into $\langle x,y\rangle$: any coefficient on a
different variable creates an uncancelled cubic monomial from
$(uy+xv+\varepsilon xy)L$.  Hence
\[
 c=\delta+\rho x+\sigma y+xy.
\]
The coefficient of $xy$ is one, so $c$ has odd parity on the four
$(x,y)$ slices.  Since $F=0$ wherever $c=0$ and the complementary
linear parts of $F$ on distinct slices differ by
$\Delta y\,u+\Delta x\,v$, $F$ cannot vanish identically on two
different slices.  Thus $c=1$ on exactly three of the four slices.

Using the seed representative \eqref{eq:seedrep}, write
\[
 M=m+Ax+By,\qquad z=n+Cx+Dy,
\]
where $m,n$ are linear forms in the six complementary variables.
Write the quadratic part of $a$ as
\[
 \lambda_0r_0+\lambda_1r_1+\lambda_\infty r_\infty.
\]
On each of the three $c=1$ slices, equality $U_s=F_s$ first gives the
complementary quadratic equation
\begin{equation}\label{eq:Qfeedback}
 m\wedge n
 =
 \lambda_1\,\bar A\wedge\bar B
 +\lambda_\infty\,a_3\wedge b_3.
\end{equation}
The four possibilities for $(\lambda_1,\lambda_\infty)$ are:
\[
\begin{array}{c|c}
(0,0)&\dim\langle m,n\rangle\le1\\
(1,0)&\langle m,n\rangle=\langle\bar A,\bar B\rangle\\
(0,1)&\langle m,n\rangle=\langle a_3,b_3\rangle\\
(1,1)&\text{impossible}.
\end{array}
\]
The last line follows because
$\bar A\bar B+a_3b_3$ has alternating rank four.

Subtracting two of the three slice equations gives the genuine
linear relation
\begin{equation}\label{eq:Lfeedback}
 \Delta\mu\,n+\Delta\nu\,m+
 \lambda_1(\Delta x\,\bar B+\Delta y\,\bar A)
 =
 \Delta y\,u+\Delta x\,v.
\end{equation}
Any three corners of $\F^2$ supply two independent differences, so
the right sides span $\langle u,v\rangle$.  In every admissible row
of \eqref{eq:Qfeedback}, the left sides lie in
$\langle m,n\rangle$.  The first row has dimension at most one; the
second is $\langle\bar A,\bar B\rangle$; the third is
$\langle a_3,b_3\rangle$.  None contains both $u=a_1$ and $v=b_1$.
Contradiction.

Crucially, the argument from \eqref{eq:Qfeedback} through
\eqref{eq:Lfeedback} did not assume
$\dim\Ann_T(h)\le2$: it applies equally to all three
dimension-three tail classes of Lemma~\ref{lem:tail}.  Thus the
proposition excludes the potentially troublesome
$g$-using first-jet case as well as the second-jet skip.
\end{proof}

\begin{remark}[why the stronger first-feedback exclusion is useful]
Lemma~\ref{lem:tail} contains six ``soft'' dimension-three annihilator
classes, namely the nonzero vectors in $P'_\infty$ and $P'_1$, in
addition to the three hard tangent classes in $P'_0$.  If a
$g$-using gate were allowed to install the first jet, those six classes
would have to be carried into the saturation analysis.  Proposition
\ref{prop:nogfirst} removes all seed-reusing gate-five branches at once.
The remaining low--low branch in Proposition~\ref{prop:firstjet} then
forces the hard tangent class $P'_0$, so Section~\ref{sec:saturation}
needs to saturate only that geometry.
\end{remark}

\subsection{The first useful gate is the first Hasse jet}

Since gate five cannot use $g$, it is a product of two low wires from
$\Aff+R$.  We can now sharpen the feedback conclusion.

\begin{proposition}[first low--low feedback]\label{prop:firstjet}
A useful gate five forces, after a $PGL_2(\F)$ change of rational
place,
\[
 h=M\wedge E_0,\qquad
 [M]\in\langle a_1,b_1\rangle\setminus\{0\}
 \quad\text{mod }P_0,
\]
and the new target direction is
\[
 E_1=D^{(1)}_0C
 \quad\text{mod }R.
\]
\end{proposition}

\begin{proof}
Compare the seed product and the low--low gate-five product using
Lemma~\ref{lem:lownf}.  Let their normal-form parameters be
$(G,N,z)$ and $(G',N',z')$.  Their cubic highs are equal and their
quadratic difference represents a nonzero target coset modulo $R$.

Suppose first $G\ne G'$.  Comparing the three direct cubic components
$I_0,I_1,I_\infty$, the only nonzero intersection patterns, up to
$S_3$ and swapping the two products, are
\[
 (G,G')=(r_0,r_0+r_1)
 \quad\text{or}\quad
 (r_0+r_1,r_0+r_1+r_\infty).
\]
In the first pattern,
\[
 N'=w\in P_1,\qquad N=w+p,\quad p\in P_0.
\]
The contraction difference in \eqref{eq:nfquad} lies in $R$, and the
quadratic difference modulo $R$ belongs to
\[
 U\wedge L,\qquad U=\langle w,p\rangle.
\]
If $p=0$ it is decomposable and cannot give a target outside $R$.
If $p\ne0$, a target outside $R$ in $U\wedge L$ has alternating rank
at most four and therefore appears in Lemma~\ref{lem:ranktwo}.
Its rank-four support would have to contain both a nonzero vector of
$P_0$ and a nonzero vector of $P_1$, which none of
$K_0,K_1,K_\infty,K_{\mathrm{deg}\,2}$ does.
The second intersection pattern has
$N=N'=w\in P_\infty$; the contraction difference lies in $R$ and the
remaining quadratic difference is decomposable.  Again no new target
is possible.

Hence $G=G'$.  If $G$ has alternating rank at least four, equality
$N\wedge G=N'\wedge G$ gives $N=N'$, and
\eqref{eq:nfquad} makes the quadratic difference decomposable modulo
$R$.  So $G$ must be a singleton rational place, say $r_0$.

Then
\[
 p=N+N'\in P_0.
\]
If $p=0$, the same decomposable-target contradiction applies.
Otherwise the quadratic difference modulo $R$ lies in
\[
 U\wedge L,\qquad U=\langle N,p\rangle,\quad\dim U=2.
\]
Any target representative outside $R$ therefore has Hankel rank two.
By Lemma~\ref{lem:ranktwo} and Lemma~\ref{lem:support}, its support
must contain $p\in P_0$.  The only outside-$R$ rank-two target forms
with this property are the two tangent forms at $0$, supported on
\[
 K_0=P_0+\langle a_1,b_1\rangle.
\]
Thus $N\in K_0$ and $N\bmod P_0$ is a nonzero vector of
$\langle a_1,b_1\rangle$.  The two tangent forms differ by $E_0\in R$,
so the new target coset is precisely $E_1+R$.
\end{proof}

Thus the state after gate five may be written
\[
 V_5=\Aff+S+\langle g\rangle,
\qquad
 S=\langle E_0,E_1,E_6,r_1\rangle,
\]
and
\begin{equation}\label{eq:hardann}
 h=M\wedge E_0,\qquad
 \Ann_T(h)=\langle E_0,E_1,E_2\rangle.
\end{equation}
The unique annihilator direction not already represented in $S$ is
\[
 E_2=D^{(2)}_0C.
\]

\section{First-order feedback cannot expose the second jet}\label{sec:saturation}

We prove that gate six cannot be useful.

\subsection{A jet-separation lemma}

Keep
\[
 x=a_0,\quad u=a_1,\quad y=b_0,\quad v=b_1,
\qquad
 K_0=\langle x,u,y,v\rangle,\quad
 P_0=\langle x,y\rangle.
\]

\begin{lemma}[jet separation]\label{lem:JS}
If
\[
 U\le K_0,\qquad \dim U\le2,\qquad U\ne P_0,
\]
then
\[
 T\cap\bigl(S+\Lambda^2K_0+U\wedge L\bigr)=S.
\]
\end{lemma}

\begin{proof}
Modulo $S$, write an arbitrary target as
\[
 A E_2+B E_3+C E_4.
\]
Suppose
\[
 q=
 AE_2+BE_3+CE_4+
 d_0E_0+d_1E_1+d_6E_6+d\,r_1
 \in\Lambda^2K_0+U\wedge L.
\]
The right side has no component in
$\Lambda^2\langle a_2,a_3,b_2,b_3\rangle$.
The coefficients of $a_2b_3$ and $a_3b_2$ give $d=0$;
the coefficient of $a_2b_2$ then gives $C=0$; and the coefficient of
$a_3b_3$ gives $d_6=0$.
Absorb the remaining $E_0,E_1$ into $\Lambda^2K_0$.

For $AE_2+BE_3$, the four $K_0$-coefficient vectors paired with
$a_2,a_3,b_2,b_3$ are
\[
 Ay+Bv,\qquad By,\qquad Ax+Bu,\qquad Bx.
\]
All belong to $U$.  If $B=1$, then $x,y\in U$, so
$U=P_0$, excluded.  Hence $B=0$.  If $A=1$, the first and third
vectors again give $x,y\in U$, the same contradiction.  Thus
$A=B=C=0$.
\end{proof}

\subsection{Low--low gate six}

\begin{lemma}[zero-wedge structure in $S$]\label{lem:Szero}
For $Q,C\in S$,
\[
 Q\wedge C=0
\]
implies either $Q,C$ are dependent or
\[
 Q,C\in\langle E_0,E_1\rangle.
\]
\end{lemma}

\begin{proof}
In the basis $(E_0,E_1,E_6,r_1)$ of $S$, the only relation among the
six pair wedges is $E_0\wedge E_1=0$; the other five are independent.
For example, on the five monomials
\[
a_0a_1b_0b_1,\quad
a_0a_1b_0b_2,\quad
a_0a_3b_0b_3,\quad
a_0a_3b_1b_3,\quad
a_0a_3b_2b_3
\]
their coefficient matrix is
\[
\begin{pmatrix}
0&1&0&0&0\\
0&1&0&1&0\\
1&1&0&0&1\\
0&0&1&1&1\\
0&0&0&0&1
\end{pmatrix},
\]
which has determinant one.  Hence the kernel of
$\Lambda^2S\to\Lambda^4L$ is the line
$\langle E_0\wedge E_1\rangle$.  If the Pl\"ucker form of
$\langle Q,C\rangle$ is nonzero and lies on that line, the two-plane
itself is $\langle E_0,E_1\rangle$.
\end{proof}

\begin{proposition}[no low--low second feedback]\label{prop:nolowsecond}
No product of two wires in $\Aff+S$ is a useful gate six.
\end{proposition}

\begin{proof}
Let the quadratic parts of the two factors be $Q,C\in S$.

First suppose the product has zero high part.  If $Q,C$ are dependent,
write them as multiples of a common $G$.  The cubic part has the form
$N\wedge G$.  If $G$ has alternating rank at least four, vanishing
forces $N=0$ and no target escapes $S$.  If $G$ is decomposable, then
by Lemma~\ref{lem:rankone} it is one of $r_0,r_1,r_\infty$.
Vanishing gives $N$ in its support plane, and the remaining possible
new quadratic term is decomposable; if it lies in $T$ it is again a
rational place and hence belongs to $S$.

If $Q,C$ are independent, Lemma~\ref{lem:Szero} puts them in
$\langle E_0,E_1\rangle\subset\Lambda^2K_0$.
Any linear component outside $K_0$ would create an uncancelled cubic
monomial because $Q,C$ are independent.  Thus the whole product is
supported in $K_0$, and
\[
 T\cap\Lambda^2K_0=\langle E_0,E_1\rangle\subset S.
\]
So a zero-high product cannot be useful.

Now suppose the product has high part $h$, so adding the seed cancels
it.  If $Q,C$ are dependent, use a common $G$.  Since
$G\wedge h=0$,
\[
 G\in S\cap\Ann_T(h)=\langle E_0,E_1\rangle.
\]
Write
\[
 M_0=\alpha u+\beta v\ne0
 \quad\text{mod }P_0.
\]
The equation $N\wedge G=h$ has the following possibilities:
\[
\begin{array}{c|c}
G&N\\ \hline
E_0&M_0+p,\quad p\in P_0,\\
E_1&\alpha x+\beta y,\\
E_0+E_1&\alpha x+\beta y.
\end{array}
\]
The quadratic shadow of the low product lies in
\[
 S+\Lambda^2K_0+N\wedge L.
\]
By the anchored normal form \eqref{eq:seedrep}, the seed shadow lies
in $S+M\wedge L$ for some representative
$M=M_0+p_0\in K_0$.  Hence the cancelled low result lies in
\[
 S+\Lambda^2K_0+U\wedge L,\qquad U=\langle M,N\rangle.
\]
Here $\dim U\le2$ and $U\ne P_0$, because $M$ has the nonzero
$M_0$ component.  Lemma~\ref{lem:JS} returns the target to $S$.

If $Q,C$ are independent, they span
$\langle E_0,E_1\rangle$.  Comparing every variable outside $K_0$
in the cubic equation forces both linear factor parts into $K_0$.
The product quadratic part is then in $\Lambda^2K_0$, and after
adding the seed the result lies in
\[
 S+\Lambda^2K_0+\langle M\rangle\wedge L.
\]
Again Lemma~\ref{lem:JS} applies, now with the one-dimensional
$U=\langle M\rangle$.
\end{proof}

\subsection{A gate-six product using the seed}

\begin{proposition}[no $g$-using second feedback]\label{prop:nogsecond}
No gate six using $g$ is useful.
\end{proposition}

\begin{proof}
Normalize the product to
\[
 F=Uc\in\Aff+T,\qquad U=g+a,\quad a,c\in\Aff+S,
\]
as before.  The identities $UF=F$ and $Fc=F$ still hold.
Degree five of $UF=F$ gives
\[
 F_2\in\Ann_T(h)\setminus S.
\]
Using \eqref{eq:hardann},
\[
 F_2=E_2+\alpha E_1+\beta E_0.
\]

Write $c=\delta+L+C$, $C\in S$.  Degree four of $Fc=F$ gives
$F_2\wedge C=0$.  Wedge by any of the four possible $F_2$ is
injective on $S$.  One explicit $4\times4$ minor, on the columns
$(E_0,E_1,E_6,r_1)$, is
\[
\begin{pmatrix}
1&0&0&1+\beta\\
0&1&0&1+\alpha+\beta\\
0&0&0&1\\
0&0&1&1
\end{pmatrix},
\]
with determinant one.  Hence $C=0$.

Degree three now gives $F_2\wedge L=0$.  The Hankel matrix
\[
 H(F_2)=
 \begin{pmatrix}
 \beta&\alpha&1&0\\
 \alpha&1&0&0\\
 1&0&0&0\\
 0&0&0&0
 \end{pmatrix}
\]
has rank three, so $F_2$ has alternating rank six and wedge by
$F_2$ is injective on vectors.  Therefore $L=0$.  Thus $c$ is
constant; $Fc=F\ne0$ forces $c=1$, whence $F=U$, contradicting the
nonzero high part of $U$.
\end{proof}

Combining Propositions~\ref{prop:nolowsecond} and
\ref{prop:nogsecond} gives the local saturation theorem.

\begin{theorem}[no second useful gate after first feedback]\label{thm:jetsaturation}
In the normalized one-defect four-term state, once the first useful
post-seed extension has been adjoined, \emph{no second useful AND
gate exists}, whether the next product is low--low or reuses the
seed.  In the rational-place coordinates of
Proposition~\ref{prop:firstjet}, the first extension is the first
Hasse jet $D^{(1)}_\theta C$; consequently the second Hasse jet
$D^{(2)}_\theta C$ is inaccessible as a special case.
\end{theorem}

\section{The unrestricted Boolean nine-AND theorem}\label{sec:main}

\begin{theorem}\label{thm:main}
The unrestricted XOR--AND multiplicative complexity of four-term
binary polynomial multiplication is
\[
 \boxed{\MC(\Mul_4)=9}.
\]
\end{theorem}

\begin{proof}
For the lower bound, assume an eight-AND circuit exists.
Lemma~\ref{lem:normalized} puts the three rational-place products in
gates one through three and the unique non-useful seed in gate four;
gates five through eight must all be useful.
Lemma~\ref{lem:highnonzero} makes the seed genuinely nonlinear, and
Proposition~\ref{prop:quartic} makes its high part cubic.
Proposition~\ref{prop:nogfirst} says gate five cannot use the seed.
Proposition~\ref{prop:firstjet} therefore forces a first-order
rational tangent and installs the first Hasse jet.
Theorem~\ref{thm:jetsaturation} says gate six cannot be useful, a
contradiction.

For the upper bound, use a two-level Karatsuba--Ofman
construction~\cite{KaratsubaOfman62}.
Let
\[
\begin{array}{lll}
g_1=a_0b_0,&g_2=a_1b_1,&
g_3=(a_0+a_1)(b_0+b_1),\\
g_4=a_2b_2,&g_5=a_3b_3,&
g_6=(a_2+a_3)(b_2+b_3),\\
g_7=(a_0+a_2)(b_0+b_2),&
g_8=(a_1+a_3)(b_1+b_3),&
g_9=(a_0+a_1+a_2+a_3)(b_0+b_1+b_2+b_3).
\end{array}
\]
Then
\[
\begin{aligned}
E_0&=g_1,\\
E_1&=g_3+g_1+g_2,\\
E_2&=g_2+g_7+g_1+g_4,\\
E_3&=(g_9+g_7+g_8)+(g_3+g_1+g_2)+(g_6+g_4+g_5),\\
E_4&=g_4+g_8+g_2+g_5,\\
E_5&=g_6+g_4+g_5,\\
E_6&=g_5.
\end{aligned}
\]
Thus nine AND gates suffice.
\end{proof}

\begin{remark}[what is new]
Kaminski's classical result and the later Hankel/code formulations give
the tight value nine in polynomial-multiplication, bilinear, and
quadratic settings~\cite{Kaminski85,FindPeralta19}.  They do not imply
Theorem~\ref{thm:main}: unrestricted Boolean circuits may multiply
nonlinear intermediate functions and exploit the identities
$x_i^2=x_i$.  Boyar--Find explicitly identify as open over $\F$ whether
quadratic-circuit lower bounds for vector-valued quadratic Boolean
functions persist for general circuits~\cite{BoyarFind18}.  Theorem
\ref{thm:main} establishes such persistence for the natural
polynomial-multiplication function $\Mul_4$.
\end{remark}

\section*{Lean formalization and reproducibility}

The complete exact theorem has a formalization in Lean~4
\cite{deMouraUllrich21}, using mathlib~\cite{Mathlib20}.  The development
represents a squarefree monomial by a finite set of variables and defines
Boolean-ANF multiplication by set union.  Simultaneous evaluation is proved
to be a linear equivalence between these canonical ANFs and all Boolean
functions.  An unrestricted circuit with $r$ AND gates is then a sequence of
products whose two factors lie in the span of the affine inputs and the
outputs of earlier gates.  Thus nonlinear feedback and Boolean idempotence
are part of the formal specification, rather than assumptions imported from
the prose proof.

The principal exported declaration is
\begin{center}
\path{UnrestrictedBooleanMul.N4.mc_mul_four : MC(Mul 4) = 9}.
\end{center}
Its lower-bound dependency chain includes the dimension obstruction for
seven gates, normalization of a hypothetical eight-gate circuit, the cubic
seed analysis, first-jet forcing, and feedback saturation; the upper bound is
an explicit nine-gate circuit.  The same project exports exact theorems for
$n=0,1,2,3$.  The pinned release uses Lean \path{v4.32.1} and a locked mathlib
revision.  A clean replay is obtained with
\begin{verbatim}
lake exe cache --cache-from=legacy get
lake build
lake env lean AxiomAudit.lean
lake env leanchecker UnrestrictedBooleanMul
\end{verbatim}
Continuous integration performs the build, axiom audit, and declaration
replay, with an additional fresh source replay on a weekly schedule.

The Lean source contains no \path{sorry}, \path{admit}, project-specific
\path{axiom}, \path{native_decide}, or \path{bv_decide}.  The explicit axiom
audit reports only Lean's standard \path{propext}, \path{Classical.choice},
and \path{Quot.sound}.  In particular, the formal proof does not import or
trust the recorded outputs of the separate Python and C++ regression checks.
The source, lock files, SHA-256 manifests, and reproduction instructions are
available at
\url{https://github.com/GregoryMorse/unrestricted-boolean-mul}; the immutable
snapshot for this manuscript is release \path{n4-arxiv-v2}, whose full commit
hash is recorded in the arXiv metadata block above.

\section*{AI assistance disclosure}

OpenAI GPT-5.6 Sol in extra-high thinking mode was used for research,
proof exploration and development, computational checking, literature and
citation verification, manuscript drafting and revision, and submission
preparation. Anthropic Opus 5 in high thinking mode was used as a referee.
The author reviewed the resulting mathematical claims, proofs, computations,
citations, code, and manuscript text and assumes full responsibility for the
final work.

\section{Toward five terms and beyond}\label{sec:n5}

For five-term inputs,
\[
 A(z)=\sum_{i=0}^4a_i z^i,\qquad
 B(z)=\sum_{j=0}^4b_j z^j,
\]
the target has dimension nine.  In the quadratic/bilinear
polynomial-multiplication model the optimum is thirteen: Montgomery gives
a division-free thirteen-product construction~\cite{Montgomery05}, and
the coding-theoretic lower bounds summarized by Find--Peralta match it
for the five-term binary case~\cite{FindPeralta19}.  The unrestricted
Boolean optimum is not inferred from that restricted result.

The rational-place prefix theorem is unchanged: over $\F$ there are
still only three rational points of $\mathbb P^1$, so every all-useful
purely quadratic prefix again stops at
\[
 R=\langle C(0),C(1),C(\infty)\rangle.
\]
What changes is the defect budget.  A hypothetical twelve-AND
computation of the nine outputs may contain three non-useful
dimensions, not one.  The n=4 proof therefore suggests the following
problem.

\begin{definition}[three-defect quadratic capacity]
For the five-term target $T_5$ and rational-place prefix $R_5$, let
$\rho_3(5)$ be the maximum target dimension recoverable when the three
non-useful quotient directions are all quadratic.  Equivalently, one
maximizes over three-dimensional quadratic defect spaces and the
decomposable fibers above their seven nonzero quotient points.
\end{definition}

A separate exact finite-geometry computation gives
\[
 \rho_3(5)=7<9.
\]
This value is \emph{not} used anywhere in the proof of
Theorem~\ref{thm:main}; unlike the $n=4$ theorem, its current status is
computational rather than fully algebraic.

\begin{conjecture}[three-defect Boolean envelope]\label{conj:multidefect}
Every five-term XOR--AND state with three total non-useful defect
dimensions satisfies
\[
 \dim(\text{recoverable target space})\le \rho_3(5)=7,
\]
even when some or all defect directions have nonzero cubic or quartic
parts.
\end{conjecture}

A positive answer would immediately rule out a twelve-AND circuit,
because such a circuit has exactly nine useful gates and three
non-useful gates, whereas all nine target dimensions would have to be
recoverable.  Thus the $n=5$ problem is not merely ``control three
high-degree directions'': the three non-useful dimensions form one
\emph{total defect budget}, mixing purely quadratic and genuinely
high-degree defects.  The structural objective is to show that
high-degree feedback cannot lift this three-defect capacity envelope.

The important issue is interaction.  One defect at one rational place
has now been analyzed completely:
\[
 C(\theta)\ \longrightarrow\ D_\theta^{(1)}C
 \quad\text{but not}\quad
 D_\theta^{(2)}C.
\]
For $n=5$, three defects might live at distinct places, or several
might interact before any one defect is saturated.  The correct
object is therefore no longer a scalar ``garbage rank.''  A more
natural invariant is a \emph{place-degree profile}: for every closed
place $P$, record the order of the jet of $C$ currently exposed there,
and charge the creation of a higher jet against the multiplicative
cost of the corresponding local algebra.

The degree-two place already appeared in the n=4 proof as the
period-three plane $D_*$.  This strongly suggests that the general
accounting should treat rational first jets and degree-two evaluation
data on the same footing.  Over $\F$, a place of degree $d$ carries
$d$ target coordinates, but multiplication in its residue algebra has
its own multiplicative cost.  The n=4 obstruction can be read as the
statement that a single rational first-order defect cannot simulate a
degree-two or second-order local multiplication for free.

A concrete n=5 program is therefore:

\begin{enumerate}[label=\textbf{N5.\arabic*.},leftmargin=3.5em]
\item Normalize the three rational-place target gates to the beginning
      of the flag, exactly as in Lemma~\ref{lem:normalized}.
\item Let $H$ be the span of the high parts of the at most three
      non-useful gates.  Classify the maps
      \[
       T_5\longrightarrow H\wedge T_5
      \]
      by rational and degree-two place support, rather than by raw
      Boolean monomials.
\item Prove a multi-defect feedback lemma: every useful gate after the
      defect space is created must expose a target class in the joint
      target annihilator of $H$.
\item Determine whether two defects at distinct rational places can
      cooperate to expose a second jet at either place.  If not, the
      n=4 saturation argument tensorizes and twelve ANDs are impossible.
\item If interaction is possible, identify the minimal local algebra
      realizing it.  That interaction, rather than another exhaustive
      lower bound, would be the genuinely new mechanism.
\end{enumerate}

The same program makes sense for arbitrary $n$.  The number of
rational places remains three, while higher-degree closed places and
higher Hasse jets enter as the target grows.  This is the main reason
to regard the n=4 proof as a structural base case rather than as a
small isolated computation.

\section{Conclusion}

The unrestricted XOR--AND flag for four-term binary polynomial
multiplication admits only a very specific kind of profitable
nonlinearity.  Purely quadratic useful prefixes are the three rational
places.  The unique high-degree defect of a hypothetical eight-gate
circuit can be normalized to a cubic rational tangent.  Its first
feedback step exposes the first Hasse derivative of the product at
that place.  Boolean idempotence and exterior jet separation then
block the second derivative.  Since an eight-gate circuit would need
four consecutive useful suffix gates, the computation cannot close.

The numerical value nine is classical for bilinear and quadratic
polynomial-multiplication models; its optimality against unrestricted
Boolean nonlinear reuse is the new lower-bound statement here.  That exact
statement, including its unrestricted semantics and full lower- and
upper-bound chain, is independently kernel-checked by the accompanying Lean~4
formalization.  The
structural message is more portable: \emph{a first-order
rational-place defect buys one first-order jet and no more}.  Whether
several such defects can interact is the next question, beginning
with five-term multiplication.

\bibliographystyle{plainnat}
\bibliography{references}

\end{document}